\documentclass[a4paper]{IEEEtran}
\usepackage{color}
\usepackage{leftidx}
\usepackage{cite}
\usepackage{microtype}
\usepackage{cleveref}
\usepackage{algorithmic}
\usepackage{amsthm}

\usepackage{amsfonts}
\usepackage{enumerate}
\newtheorem{lemma}{Lemma}
\usepackage{mathtools}
\usepackage{diagbox,booktabs}
\usepackage[justification=centering]{caption}
\usepackage{geometry}
\usepackage{graphicx,amssymb,amstext,amsmath}
\usepackage[ruled,vlined,lined,ruled,linesnumbered]{algorithm2e}
\usepackage{float}
\usepackage{makeidx}

\begin{document}

\title{Throughput Maximization for Ambient Backscatter Communication: A Reinforcement Learning Approach}
\author
{\IEEEauthorblockN{Xiaokang Wen, Suzhi Bi$^{*}$, Xiaohui Lin, Lina Yuan, and Juan Wang}\\
    \IEEEauthorblockA
    {
    College of Information Engineering, Shenzhen University\\
    Shenzhen, Guangdong, 518060, China\\
    Email: wenxiaokang2016@email.szu.edu.cn, \{bsz, xhlin, yln, juanwang\}@szu.edu.cn
    }
}

\maketitle
\thispagestyle{empty}

\begin{abstract}
Ambient backscatter (AB) communication is an emerging wireless communication technology that enables wireless devices (WDs) to communicate without requiring active radio transmission. In an AB communication system, a WD switches between communication and energy harvesting modes. The harvested energy is used to power the devices operations, e.g., circuit power consumption and sensing operation. In this paper, we focus on maximizing the throughput performance of AB communication system by adaptively selecting the operating mode under fading channel environment. We model the problem as an infinite-horizon Markov Decision Process (MDP) and accordingly obtain the optimal mode switching policy by the value iteration algorithm given the channel distributions. Meanwhile, when the knowledge of channel distribution is absent, a Q-learning (QL) method is applied to explore a suboptimal strategy through device repeated interaction with the environment. Finally, our simulations show that the proposed QL method can achieve close-to-optimal throughput performance and significantly outperforms the other than representative benchmark methods.
\end{abstract}
\vspace{-2ex}
\IEEEpeerreviewmaketitle

\section{Introduction}
The future Internet of thing (IoT) technology interconnects numerous sensing devices with communications capability for a wide range of applications, e.g., remote monitoring, automatic control, diagnosis and maintenance \cite{2010:Tan}. Recently, a new communication paradigm named ambient backscatter (AB) communication is widely studied as an energy-efficient method applicable in IoT system \cite{2013:Liu}. In particular, a tag transmitter in AB communication system communicates with its receiver by backscattering its ambient radio frequency (RF) signals. Specifically, a transmitter tag transmits `0' or `1' by switching its antenna to non-reflecting or reflecting mode, respectively. Compared to the conventional backscatter communication scheme in radio frequency identification (RFID) systems, AB communication does not require a dedicated energy-emitting reader, and relies solely on external energy sources in the ambient environment, such as WiFi, public radio, and cellular transmit power. As such, the application of AB communication can effectively reduce the deployment cost of large-size IoT network, such as smart homes, smart cities, and environment monitoring \cite{2017:Khan}, \cite{2017:Hui}, \cite{2017:Alsinglawi}.

There has been tremendous research interests recently on ambient backscatter communications \cite{2013:A}, \cite{2013:M}. For instance, \cite{2015:Lu} analyzed the bit error rate of an AB communication link when the receiver uses an energy detector to detect the 1-bit information transmitted per channel use. \cite{2017:Hoang} integrates the AB communication with conventional harvest-then-transmit (HTT) protocol in the radio frequency-powered cognitive radio networks, where the backscatter tag can choose to backscatter the ambient RF signal to the receiver or harvest energy for later active transmissions. To achieve the optimal throughput performance, the authors assume a fixed channel model and optimize the time allocation on backscattering communication, energy harvesting, and active information transmissions. AB communication has also been integrated in wireless powered communication network, where the wireless devices' information transmissions are powered by means of wireless power transfer \cite{2016:Bi}, \cite{2018:Bi}. For instance, \cite{2018:Zheng}, \cite{2018:Xu} consider using backscatter communication to reuse wireless power transfer for simultaneous energy harvesting and information exchange between two cooperating users, It shows that the use of passive backscatter-assisted cooperation can significantly improve system throughput performance compared to conventional active information transmissions.

The above studies mostly focus on the system performance optimization under given channel state within a time slot, while the channel fading effects across consecutive wireless channels are not considered. In practice, the wireless channel fading cause the ambient signal strength to vary over time, which directly results in a time-varying communication performance. In general, the choice of current operating mode, i.e., backscattering communication or harvesting energy, to maximize the data rate depends on several factors, such as the current channel conditions, battery energy, and the circuit consumption, etc. However, the dynamic operating mode selection problem in fading channel environment has not been well addressed so far.

In this paper, we concentrate on maximizing the long-term average throughput of an AB communication system by optimizing the real-time operating mode strategy of a transmitter tag in fading channel. Particularly, we model it as an infinite multi-stage decision problem and formulate an infinite-horizon Markov Decision Process (MDP) problem. When the channel distribution is known, we apply the value iteration method to obtain the optimal decision strategy. In practice, however, such channel state knowledge is often hard to obtain, and accordingly, we propose a Q-learning (QL) based reinforcement learning method that obtains a sub-optimal strategy without knowing the channel distribution. Simulation results show that the QL method produces close-to-optimal throughput performance, and significantly outperforms the other representative benchmark methods.

\section{System Model}
\subsection{Channel Model}
As shown in Fig.~\ref{Fig.1}, we consider an AB communication system consisting of one RF source, and a pair of AB transmitter and receiver, all of which are equipped with single antenna. All the channels are assumed to follow quasi-static flat-fading, such that all the channels coefficients remain constant during each block transmission time $T_0$, but can vary from different blocks. The channel coefficients, between the RF source and the tag, between the RF source and the receiver, and between the tag and the receiver, are denoted by $\alpha_{st}$, $\alpha_{sr}$, and $\alpha_{tr}$, respectively. Correspondingly, we use $g=|\alpha_{st}|^2$, $f=|\alpha_{sr}|^2$, and $h=|\alpha_{tr}|^2$ to denote their channel gains, separately, where $|.|$ represents the 2-norm operator.
\begin{figure}
  \centering
   \begin{center}
      \includegraphics[width=0.45\textwidth]{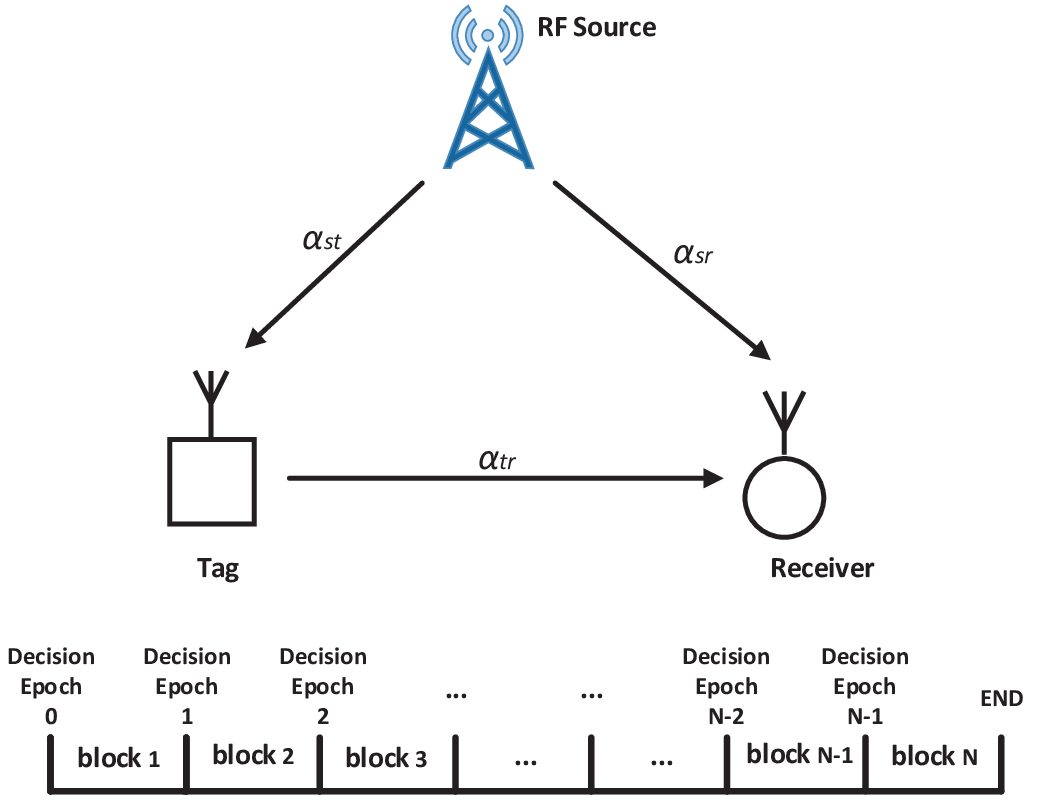}
   \end{center}
  \caption{The considered AB communication system. }
  \label{Fig.1}
\end{figure}

We consider consecutive decision epochs in Fig. 1, where two adjacent epochs are separated by equal duration $T_0$. At the $t$-th transmission block, the received signals at the tag can be expressed as
\begin{equation}
 \label{yb}
 y_b(t) = {\alpha_{st}}{x(t)+w(t)},
\end{equation}
where $x(t)$ is the RF signal transmitted from the ambient RF source, $\alpha_{st}$ denotes the channel power gains between RF source and backscatter-tag, and $w(t) \sim \mathcal{CN}(0, \sigma_w^2)$ denotes the additive white gaussian noise (AWGN) between RF source and tag. At the beginning of each epoch, the tag makes a decision on either operating at signal-backscattering mode or energy-harvesting mode. The circuit block diagram of the tag is illustrated in Fig.~\ref{Fig.2}, $S_1$ and $S_2$ can switch the connection point to change their operating mode in real-time.
\begin{figure}[!t]
 \centering
  \begin{center}
      \includegraphics[width=0.45\textwidth]{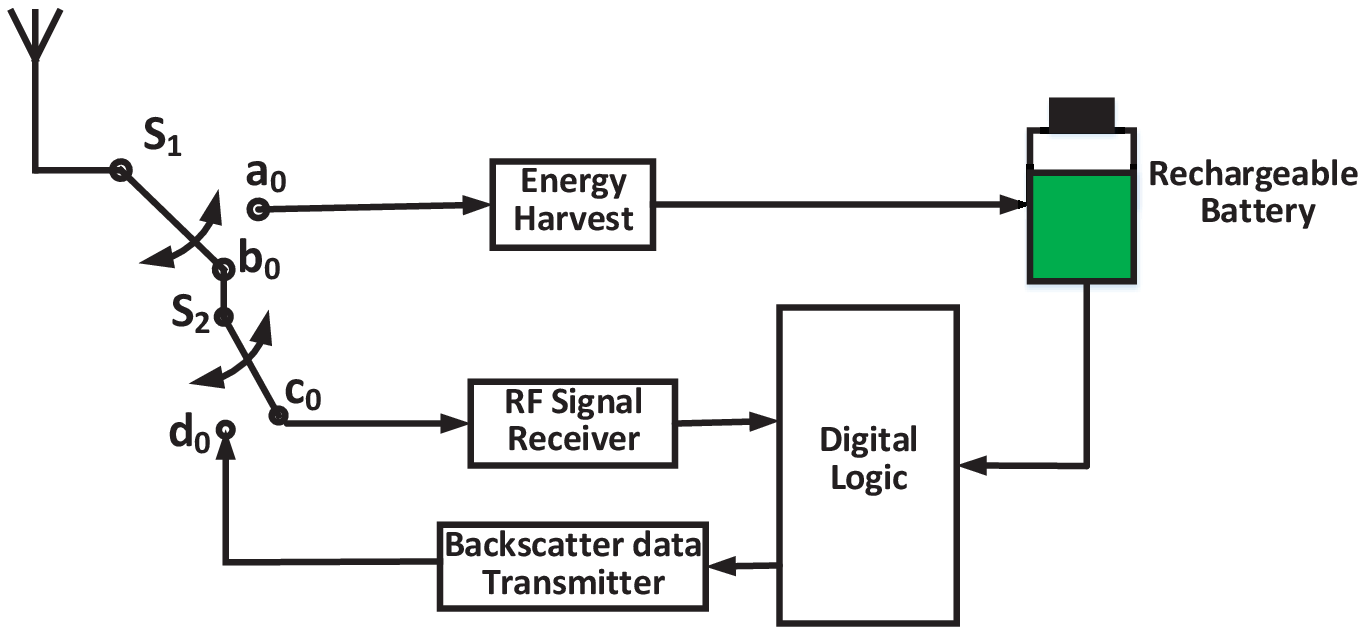}
  \end{center}
 \caption{The circuit block diagram of the tag.}
 \label{Fig.2}
\end{figure}

When the switch $S_1=a_0$, the tag operates in energy harvesting mode. The energy harvesting circuit converts RF signal into direct current (DC) power to charge the battery. The collected energy is used for data transmission or replenishing circuit consumption. The harvested energy can be expressed as
\begin{equation}
 \label{Eg}
 E_t = {\eta}{g_t}{P_t}{T_0},
\end{equation}
where $\eta$ is the battery energy harvesting efficiency and $P_t$ is the fixed power of RF source. $g_t$ denotes the channel power gain in the $t$-th time slot. For simplicity of illustration, we consider a truncated channel gain (e.g, $99\%$ cumulative distribution) $g=|\alpha_{st}|^2$ and quantize it into ($Y+1$) levels $\mathcal{G} = \left\{G_0,G_1, G_2,\cdots, G_Y\right\}$. Therefore, the possible harvested energy can be divided into ($Y+1$) uniform levels, such that
\begin{equation}
 \label{E}
 E_{g(t)} \in \left \{0,1,2,\cdots,Y \right \}\cdot e_0,
\end{equation}
where, $e_0$ denotes the unit energy considered for quantization. Notice that the tag may harvest zero energy when the received signal is too weak.

When $S_1=b_0$ and $S_2=d_0$, the tag switches to signal-backscattering mode. In this case, the energy collected by the tag is approximately zero. The received signal at the receiver, as a combination of signal transmitted by the RF source and backscattered by the tag, is
\begin{equation}
\label{yn}
 \begin{aligned}
  y_r(t) = {\mu}{\alpha_{st}}{\alpha_{tr}}{a(t)}{x(t)}+{\alpha_{sr}}{x(t)}+{w(t)},
 \end{aligned}
\end{equation}
where $\mu$ is the reflection coefficient at the tag, $\alpha_{tr}$ is the channel coefficient from the tag to the receiver that remain fixed in the considered period and $a(t)$ denotes the decision of a backscatter tag in the $t$-th time slot. Generally, the distance from RF source to tag and the distance from RF source to receiver are much larger than the distance between the tag and receiver. We therefore assume that the received signal strengths at the tag and receiver are the same, i.e., $g=f$.

We assume the tag transmits with a fixed data rate $R_b$ bits per second and the sampling rate of the receiver is $N_sR_b$, such that the receiver takes $N_s$ samples of every one-bit transmissions. In the following, we derive the BER of the receiver using an optimal energy detector to decode the received information.
\begin{lemma}
 \rm Let $\delta_0^2$ and ${\delta_1}^2$ represent the variance of an addition noise introduced by the receiver RF circuit and the decoding circuit, respectively. Using an optimal energy detector, denote the BER at the receiver $\epsilon$ can be expressed as
 \begin{equation}
  \label{ber}
  \epsilon = \frac{1}{2}erfc\left [ \frac{(\mu ^2P_{t}gh\sqrt{N_s})}{4({\delta_{0}}^{2}+ {\delta_{1}}^2) }\right ],
 \end{equation}
\end{lemma}
\begin{proof}
 Please refer to Appendix A.
\end{proof}

We denote the BER in the $t$-th time slot as $\epsilon(t)$. Then, the capacity of the binary symmetric channel is
\begin{equation}
\label{C}
 C(t) = 1 + {\epsilon(t)}log(\epsilon(t)) + (1-\epsilon(t))log(1-\epsilon(t)).
\end{equation}
Therefore, the data rate of the backscatter communication in the current time slot is
\begin{equation}
\label{R}
 R(t) = {R_b}{C(t)}{T_0}.
\end{equation}

\subsection{Battery Model}
We quantize the battery capacity $C$ by $e_0$ into $B_{c}$ units, where $B_{c}= C/e_0$ is assumed without loss of generality to be an integer. The tag consumes $j$ units energy for maintaining the basic energy consumption of the circuit when operating on the energy-harvesting mode and $k$ units of energy in signal-backscattering mode, where $1\leq j < k < B_{c}$. At the beginning of epoch $t$, the tag can operate on the signal-backscattering mode only when the energy state $E_c(t)\geq k$. Otherwise, it must harvest enough energy by operating in the energy harvesting mode. We let $a(t)$ denote the operating mode selection, where $a(t)=0$ indicates energy harvesting mode and $a(t)=1$ otherwise. Accordingly, the dynamic of the battery energy $E_c(t+1)$ can be expressed as
\begin{equation}
 \label{Ec}
 E_c(t+1)\!\!=\!\!\left\{
 \begin{aligned}
 &{E_c(t)\!-\!j\!+\!E_g(t)},\!\!\!\!\!\!\!\!\!\!\!\; &{if E_c(t)<k},\\
 &{E_c(t)\!-\!ka(t)\!+\!(1\!-\!a(t))E_g(t)\!\!-\!\!j},\!\!\!\; &{if E_c(t)\geqslant k},
\end{aligned}
\right.
\end{equation}
for $t=0,1,\cdots,N$, where $E_c(0)=E_0$ represents the initial status of the tag battery.

\subsection{Problem Formulation}
As shown in Fig.~\ref{Fig.1}, we intend to maximize the long-term throughput of a tag in a very large number of $N$ time slots. Here, we use $\pi$ to represent a static decision strategy in choosing the operating mode. $R^\pi(t)$ and $E^\pi_c(t)$ denote the achievable data rate and battery energy state as a result of the strategy $\pi$ at the $t$-th time slot. The objective is to find an optimal policy $\pi^*$ to maximize the average throughput. Mathematically, the problem can be formulated as
\begin{equation}
\label{p1}
\begin{aligned}
 & R^{\pi^*}=&&\underset{\pi}{max}\frac{1}{N}\lim_{t\rightarrow \infty }\sum_{t=0}^{N-1}\gamma ^{t}R^{\pi}(t),\\ 
 &s.t. &&(8)~and~E^{\pi}_c(0) = E_0, \\
 & &&0\leq {E^{\pi}_{c}(t)}\leq {B_c}. \\
\end{aligned}
\end{equation}
where $\gamma \in (0,1]$ is the discount factor.

\section{Reinforcement Learning Approach}
\subsection{Markov Decision Process}
Depending on the knowledge of the distribution of ambient RF signal strength, we propose in this section to solve (\ref{p1}) using both optimal model-based Markov decision process method and model-free reinforcement learning method.  When the distribution of the ambient RF signal strength follows a Markov process and is known, the discrete time-slots decision problem in (\ref{p1}) can be described as an MDP. In the following, we define the five major elements of an MDP for solving (\ref{p1}): states ($S$), actions ($\mathcal{A}$), transition probability ($P(s'| s, a)$), immediate reward ($R(s'| s, a)$), and discounter ($\gamma$). First of all, we define a state $s\in \mathcal{S}$ by the $E_c$ unit(s) of current battery energy and the channel gain $g$ in a decision epoch. That is, $S = \left\{(E_c, g),  E_c\in \left\{0,1,...,B_c\right\}e_0, g\in \mathcal{G} \right\}$. Because there are in total $(Y+1)$ discrete energy state, the cardinality of the state space is $|\mathcal{S}| = (B_{c}+1)\times (Y+1)$. The tag takes an action on choosing either energy-harvesting mode or signals-backscattering mode at every decision epoch, which are described as follows
\begin{enumerate}
  \item \emph {Actions} : a backscatter tag adaptively switches between energy-harvesting mode and signal-backscattering mode based on its current state. Let $\mathcal{A} = \left\{0, 1\right\}$ represent the action set, where $a = 0$ and $a = 1$ denote energy-harvesting and signal-backscattering mode, respectively.
  \item \emph {Rewards} : we define the immediate reward received by the tag as $R(s'|s,a)$ as the amount of information successfully transmitted to the receiver. Here $s$ and $s'$ denote the current state and the state of the next decision epoch. With a bit abuse of notation, we denote $C(s)$ as the channel capacity when the system is at state $s$. Then, the reward is
      \begin{equation}
      \label{rwd}
      R(s'|s,a) = a{R_b}{C(s)}{T_0}.
      \end{equation}
      Notice that a tag may receive immediate reward only when operating in signal-backscattering mode ($a=1$). Operating in energy-harvesting mode (a=0) has no immediate reward, but the energy collected at the current slot can be used to support data transmission in latter slots.
  \item \emph {Transition Probabilities} : the channel state transition probability is assumed to be static throughout all the time slots. We define transition probability matrix $P$ with its elements $P_{ij}=P(s_j| s_i, a)$, as the probability of transiting to $s_j$ when taking an action $a$ at sate $s_i$. With random energy arrival $E_g(t)$, the battery state has been given in (\ref{Ec}). For each state-action pair $(s_i,a)$, it satisfies
  \begin{equation}
  \sum_{j = 1}^{n}{P(s_j| s_i, a)} = 1, \forall~s_i \in \mathcal{S}.
  \end{equation}
\end{enumerate}

We aim to find an optimal policy $\pi^{*}(s_i) \in \mathcal{A}$ for every state $s_i \in \mathcal{S}$, which maximizes the average throughput reward over a long time. Based on the knowledge of transition probability, we can get the global optimal policy with the value iteration algorithm, which is one widely used algorithm for solving discounted MDP problems \cite{1957:Bellman} and detailed as follows.

The value iteration algorithm aims to estimate the expected reward received at each state $s$, denoted by $V(s)$, for all $s\in \mathcal{S}$. In particular, iteration starts with setting $V(s)=0$, for all $s \in \mathcal{S}$, and chooses the next state by taking an local optimal action $a^*$ that maximizes its expected reward in the current stage. The two action in each state $s$ during a iteration will have different immediate reward $R(s'|s,a)$ and discounted future reward. At the end of each iteration, we select the best action for each state and update the reward function $V(s)$ for each state. Overall, the value function $V(s)$ is updated by
\begin{equation}
\label{VI}
V(s) :=\!\underset{a \in \left\{0,1\right\}}{max}[R(s'| s, a)\!+\!\gamma\sum_{s^{'}}P(s^{'}|s,a)V(s^{'})],  \forall s \in \mathcal{S}.
\end{equation}

The iterations proceed until the maximum difference among all the states between two consecutive iterations is less than a certain threshold $\theta$, i.e.,
\begin{equation}
\label{v}
\underset{s \in \mathcal{S}}{max}\left | V^{(l)}(s)-V^{(l-1)}(s)\right | < \theta , \forall s \in \mathcal{S},
\end{equation}
where the superscript denotes the $l$-th iteration. The convergence of the algorithm is guaranteed when sufficient number of iterations are taken \cite{2001:Lzhang}. We denote the value function after convergence as $V^*(s)$. Then, the optimal strategy $\pi_{v}^*$ of value iteration algorithm is therefore
\begin{equation}
\label{vlb}
\pi_{v}^*(s) = {arg} \underset{a \in \left\{0,1\right\}}{max}[R(s'| s, a)\!+\!\gamma\sum_{s^{'}}P(s^{'}|s,a)V^*(s^{'})]
\end{equation}

Through extensive experiments, we observe for the optimal policy that, when the battery energy is low but the channel condition is good, it tends to harvest energy for transmitting data in latter slots. Conversely, when the battery energy is high, it is inclined to transfer data to consume energy to avoid the harvest energy overcharging the battery. When the battery energy of the tag is moderate, it chooses to harvest energy when the channel condition is poor and transmits data when the channel conditions are relatively good. The detailed performance of the value iteration method will be shown in simulations.

\subsection{QL Algorithm}
When the distribution of the ambient RF signal strength is not known, we consider using a QL based online algorithm to find a suboptimal mode selection strategy. In each decision epoch, the tag chooses an action based on the Q-value in a constructed state-action value table subject to constant update upon iterative interactions with the environment. The table is initialized by setting $Q(s,a)=0$, for all states $s \in \mathcal{S}$ and actions $a \in \mathcal{A}$. The iterations start by picking a random state $s\in \mathcal{S}$. To update Q-table, $\epsilon-$greedy method is used for balancing exploration and exploitation. With $\epsilon-$greedy, the tag selects a random action with probability $\epsilon_0$, where $0 \leq \epsilon_0 \leq 1$ and reduces over time. Then, with probability (1-$\epsilon_0$), the entry corresponds to the $(s,a)$ in the Q-table is updated by
\begin{equation}
\label{QL}
Q(s,a) := Q(s,a) + \alpha \underset{a'}{max}[R(s'| s, a) + \gamma Q(s',a') - Q(s,a)],
\end{equation}
where $\alpha$ is a small learning rate.
The tag takes the action $a$ that maximizes (\ref{QL}) and receives immediate reward  $R(s'|s,a)$ in (\ref{rwd}) if $a=1$. After taking the action, the tag observes the next state $s'$ following (\ref{Ec}) and the unknown channel transition probability.

The tag will make better mode selection over time, and after sufficiently long learning period the values in the Q-table will stabilize. We use $\pi_{q}^*$ to represent its the mode selection of QL algorithm after the Q-table stabilizes. The details of QL algorithm is showed in Algorithm 1.
\begin{algorithm}
\caption{QL algorithm in mode selection}
\label{alg}
\KwIn{$S$: All states; $R(s'| s, a)$: Immediate reward matrix; $R_b$: Data rate; $\epsilon_0$: exploration probability; $\gamma$: Discount factor;}
\KwOut{$\pi_{q}^*$;}
Initialize $Q(s, a) = 0$ for all $s \in \mathcal{S}$ and $a \in \mathcal{A}$;
$\alpha \gets $a small learning rate\;
Randomly select a beginning state $s$\;
\For {$t:=1$~\rm to~$\infty$}{
    On the current state $s$, generate a probability $p_0 \sim \mathcal{U}(0,1)$\;
¡¡¡¡\eIf{$p_0 \leq {\frac{1}{\sqrt{t}} \cdot \epsilon_0}$}{take a random action $a(s)=0$ or 1 with equal probability\;}
    {$a(s)$ = {arg} $\underset{a \in \left\{0,1\right\}}{max} Q(s,a)$\;}
    Take action $a(s)$ and observe the reward $R(s'| s, a)$ and next state $s'$\;
    Update the Q-table values using (\ref{QL})\;
    Update $s \gets s'$.
}
\For{each $s \in \mathcal{S}$}
 {
    ${a^{*}(s)} = {arg} \underset{a \in \left\{0,1\right\}}{max} Q(s,a)$\;
 }
\Return $\pi_{q}^*=\left\{a^{*}(s_1),a^{*}(s_2),\cdots,a^{*}(s_{|\mathcal{S}|})\right\}$.
\end{algorithm}

\section{SIMULATION RESULTS}
In this section, we evaluate the performance of the proposed algorithms. In all simulations, we assume the energy harvesting efficiency $\eta=0.8$ and the exploration probability $\epsilon_0=0.2$. Without loss of generality, the duration of each block transmission time $T_0$ is set to 1. In addition, the noise power is assumed as $N_0= 10^{-10}~W$. We set $Y=4$ and the channel levels {$G_0$, $G_1$, $G_2$, $G_3$, $G_4$} are set as $1.5\times10^{-5}$, $3\times10^{-5}$, $4.5\times10^{-5}$, $6\times10^{-5}$, $7.5\times10^{-5}$, respectively. Unless otherwise stated, $h$ is fixed to be a constant during all time slots, i.e., $h=5\times10^{-5}$. $R_b$ equals to $10$ Kbits per second. We set the unit energy as $e_0=\eta G_1 P_t T_0$. Therefore, the harvested energy after quantization is $E_g = i$ units energy when $g=G_i$, where $i=1, \cdots, 4$, and $E_g=0$ if $g=G_0$. Without loss of generality, we assume battery capacity $B_c=9$. Besides, the tags consumes $j=1$ and $k=3$ units of circuit energy when operating in energy harvesting and signal backscattering mode, respectively. The channel transition probability from $G_i$ to $G_j$ in consecutive time slots is denoted as $P_{ij}$, and is showed in Table.~\ref{tab1}.
\begin{table}[!t]
 \begin{center}
  \caption{Transition Probability Matrix}
 \label{tab1}
 \begin{tabular}{c|c|c|c|c|c}
 \toprule
 \diagbox{i}{j} & G0 & G1 & G2 & G3 & G4 \\
 \hline
  G0 & 0.40 & 0.30 & 0.15 & 0.10 & 0.05 \\
 \hline
  G1 & 0.05 & 0.40 & 0.30 & 0.15 & 0.10 \\
 \hline
  G2 & 0.10 & 0.05 & 0.40 & 0.30 & 0.15 \\
 \hline
  G3 & 0.15 & 0.10 & 0.05 & 0.40 & 0.30 \\
 \hline
  G4  & 0.30 & 0.15 & 0.10 & 0.05 & 0.40 \\
 \bottomrule
 \end{tabular}
 \end{center}
\end{table}

As a benchmark method for performance comparison, we consider a greedy policy. Specifically, the tag chooses signal-backscattering mode if it has sufficient energy to transmit information, or energy-harvesting mode otherwise. That is,
\begin{equation}
\label{greedy}
a(t)= \begin{cases}
&0, \text{ if } E_c(t) < k, \\
&1, \text{ if } E_c(t) \geq k.
\end{cases}
\end{equation}
\begin{figure}
 \centering
 \begin{center}
    \includegraphics[width=0.50\textwidth]{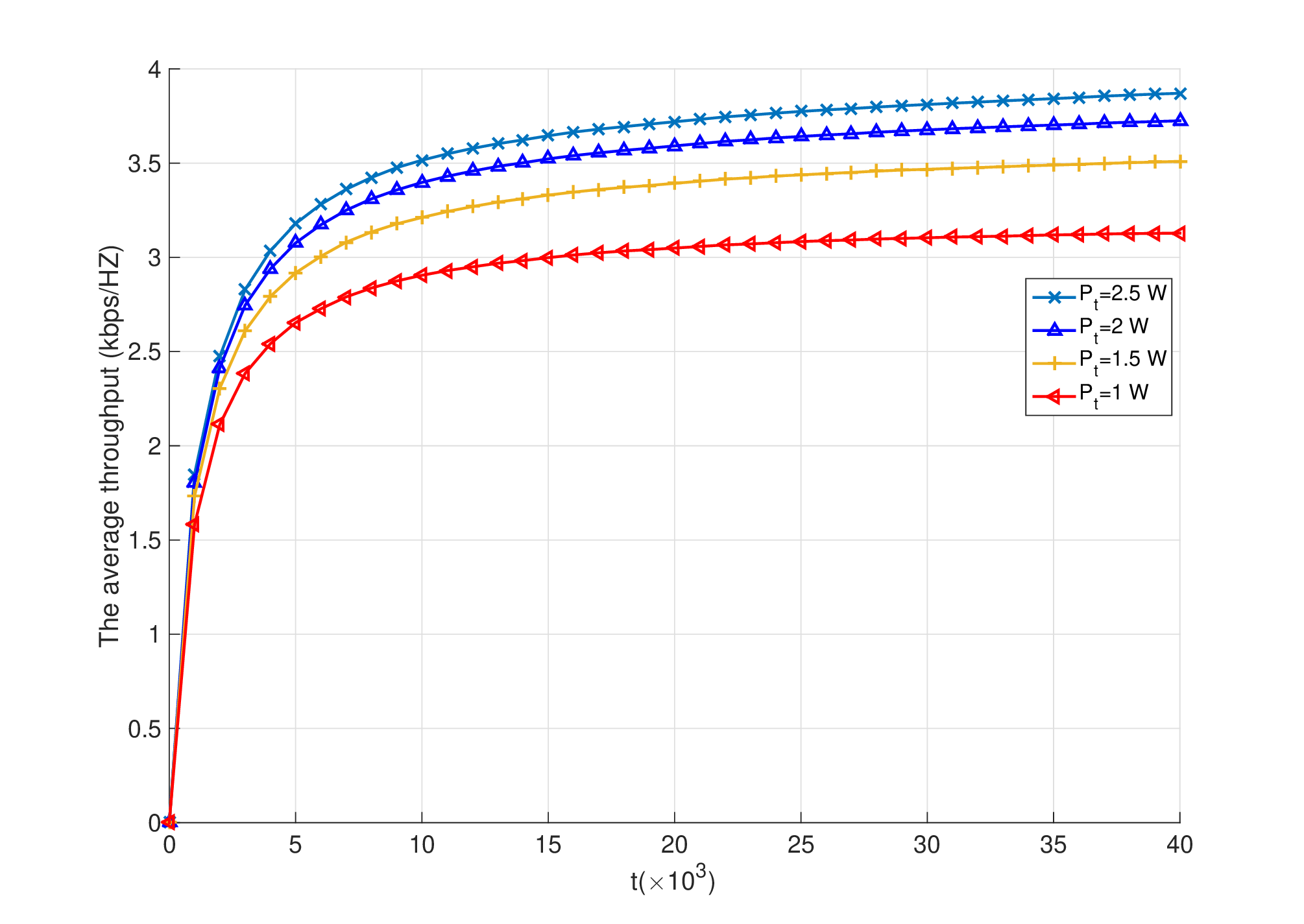}
 \end{center}
 \caption{Throughput performance variation with the number of iterations of the QL method.}
 \label{Fig.3}
\end{figure}
\begin{figure}
 \centering
 \begin{center}
    \includegraphics[width=0.50\textwidth]{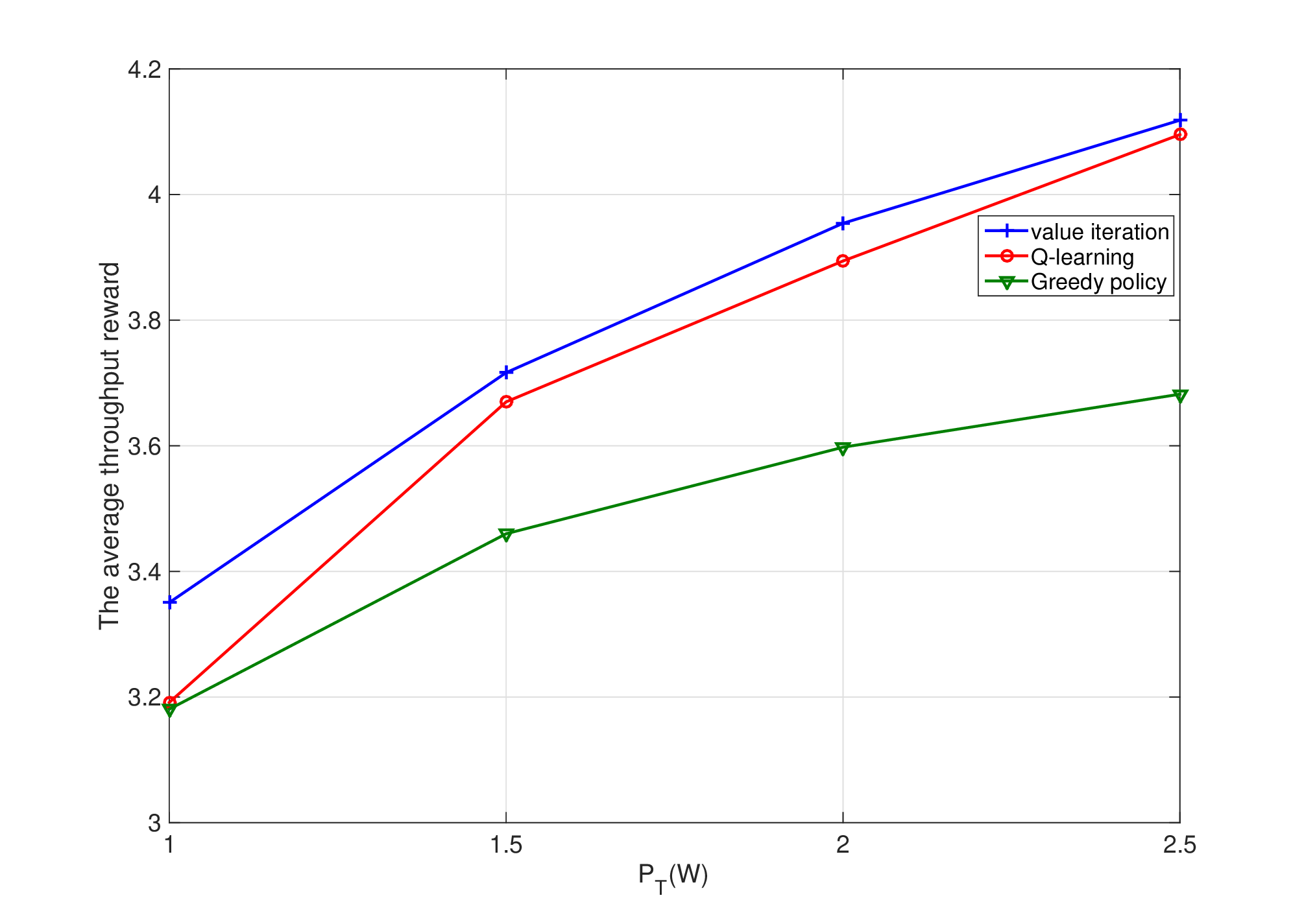}
 \end{center}
 \caption{Throughput performance comparisons of the different methods as a function of RF source power.}
 \label{Fig.4}
\end{figure}
\begin{figure}
 \centering
 \begin{center}
    \includegraphics[width=0.5\textwidth]{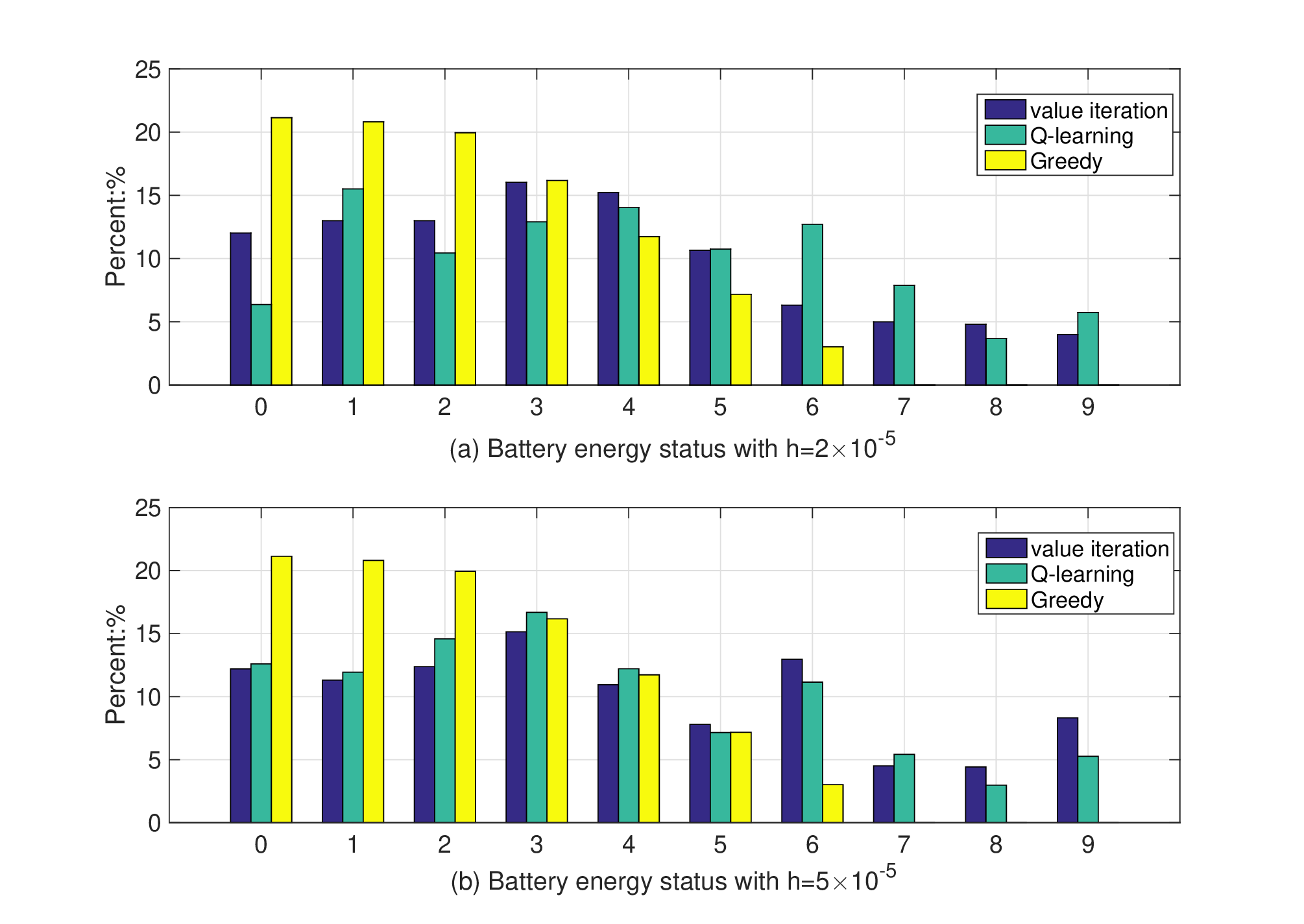}
 \end{center}
 \caption{The distribution of battery energy level in $10000$ simulation time slots. (a) channel power gain $h=2\times10^{-5}$, (b) $h=5\times10^{-5}$.}
 \label{Fig.5}
\end{figure}
We first show in Fig.~\ref{Fig.3} the average throughput achieved by the QL algorithm as the number of iterations. Each point in the figure is a rolling average of the past $10^3$ time slots. The four curves from bottom to top represent the cases where the power of the RF source increases from $1~W$ to $2.5~W$. As expected, a higher source power leads to higher average throughput of the energy harvesting device. Besides, we see that the average throughput performance under different transmit power gradually increases as the iterations proceed, and saturates at around $20\times 10^3$ iterations. In other words, the tag makes better mode operating decisions over time, and the Q-table eventually becomes stable after sufficiently long interaction with the environment.

In Fig.~\ref{Fig.4}, we compare the average throughput performance of three methods: the value iteration, QL and greedy algorithms, when the source transmit power varies from $1$ to $2.5$ Watts. In particular, for the value iteration and QL methods, we use the mode selection strategies after both methods converge. For fair comparisons, we evaluate the three methods in $N=10^4$ time slots, where the channel realizations follow that in Fig.~\ref{Fig.3}. Each point in the figure is the average throughput achieved within the $N$ time slots. It is evident in Fig.~\ref{Fig.4} that the average throughputs increase with $P_t$. The QL algorithm achieves very close throughput performance to the optimal value iteration algorithm, and significantly outperforms the greedy method. Specifically, the performance loss is less than 0.56\% when the transmitter power $P_t=2.5~W$. On average,
the QL method achieves 98.17\% of the optimal throughput performance, and the greedy method achieves 90.33\% of the optimal performance, while the performance gap of the greedy method gradually increases as $P_t$ becomes larger. The throughput performance of the greedy method is worse than that of the other two because it only considers maximizing the current reward while neglecting the significant future reward achievable by operating in energy-harvesting mode. The QL method, although has no knowledge of the channel distribution, achieves close-to-optimal performance when transmit power is large.

In Fig.~\ref{Fig.5}, we simulate the performance of the three mode selection methods in  $10^4$ time slots, and plot the probability distribution of the battery energy levels of the tag during the entire simulation. Here, we consider two different tag-to-receiver channel conditions $h=2 \times 10^{-5}$ and $h=5\times 10^{-5}$. In both cases, we can see that the greedy method results in low battery energy states in both cases, where more than 80\% of the time the tag has less than $3$ units of energy left in the battery and has not even reached $7$ units of energy throughout the simulation. This is due to its greedy nature in exhausting any energy available, thus transmission outage happens frequently when a favorable transmission opportunity occurs, resulting significant loss of data rate. Conversely, the optimal value iteration algorithm results in a much more balanced battery energy distribution in different energy states, such that it leaves sufficient ``energy buffer" for transmitting information when a favorable slot occurs, the QL algorithm closely follows the energy distribution of value iteration algorithm, which shows its ability to jointly consider both current and future data transmission opportunities.

\section{CONCLUSION}
In this paper, we studied the optimal operating mode selection problem in the AB communication system, where the backscattering tag dynamically chooses between energy-harvesting and information backscattering modes to maximize the average throughput. We formulated the problem into an infinite-horizon MDP problem. When the the distribution of the ambient RF signal strength is known, we applied value iteration algorithm to find the optimal decision strategy. Otherwise, when the signal strength distribution is not known, we proposed to employ reinforcement QL algorithm to maximize the long-term average throughput. Finally, our simulations showed that the proposed QL method can achieve close-to-optimal throughput performance and significantly outperforms the benchmark greedy method in the AB communication system.

\section*{Appendix A\\Proof of lemma 1}
Let $D[i]\in \left\{0,1\right\}$ denotes the information bit transmitted in the current time slot, the received signal at the receiving end in backscatter communication system, $y_r[i]$ can be expressed as
\begin{equation}
y_r[i]=\alpha_{sr}x[i] + D[i]\mu \alpha_{st}\alpha_{tr}x[i]+n_0[i],i=1,\cdots,N_s.
\end{equation}
where $D[i]$ denotes the binary information bits, $n_0 \sim \mathcal {CN}(0 ,N_0)$, and the signal at information decoder is
\begin{equation}
y[i]= y_r[i]+n_1[i],i=1,\cdots,N_s,
\end{equation}
where $n_1 \sim \mathcal {CN}(0 ,N_1)$, the average power harvested in the corresponding $N_s$ symbol is
\begin{equation}
E[\frac{1}{N_s} \sum_{i=1}^{N_s}|y[i]|^2] = (P_t|\alpha_{sr}+D[i]\mu\alpha_{st}\alpha_{tr}|^2+N_0)+N_1.
\end{equation}

It can be clearly shown that the following equalities hold
\begin{equation}
E[\sum_{i=1}^{N_s}[n_1[i]^2] = {N_s}{N_1},\quad
Var[\sum_{i=1}^{N_s}[n_1[i]^2] = 2{N_s}{N_1^2}.
\end{equation}
When $N_s$ is sufficiently large, by the central limit theorem, the test statistic $Z = \frac{1}{N_s} \sum \limits_{i=1}^{N_s} |y[i]|^2$ for both cases can be expressed as
\begin{equation}
\label{twodis}
 \begin{aligned}
&D[i]=0\!:Z\!\sim \mathcal {N_s}(P_th\!+\!N_0\!+\!N_1, \frac{2{(N_0\!+\!N_1)}^2}{N_s}), \\
&D[i]=1\!:Z\!\sim \mathcal{N_s}(P_t|\alpha_{sr}\!+\!\mu\alpha_{st}\alpha_{tr}|^2\!+\!N_0\!+\!N_1,\frac{2{(N_0\!+\!N_1)}^2}{N_s}).
 \end{aligned}
\end{equation}

By defining $Z_1 = Z - P_tf-N_0-N_1$, we have
\begin{small}
\begin{equation}
\label{twodis}
 \begin{aligned}
D[i]=0:Z_1 \sim \mathcal {N_s}(&0,\frac{2(N_0+N_1)^2}{N_s}),\\
D[i]=1:Z_1 \sim \mathcal{N_s}(&P_t|\mu^2gh + 2\mu\alpha_{st}\alpha_{sr}\alpha_{tr}|, \frac{2(N_0+N_1)^2}{N_s}).\\
 \end{aligned}
\end{equation}
\end{small}
We assume that `0' and `1' are transmitted with equal probability. Thus, the bit error probability (BER) $\epsilon$ can be obtained as
\begin{equation}
\begin{aligned}
 \epsilon &=\frac{1}{2}(P_r(\hat{D}(i)=0|D(i)=1)+P_r(\hat{D}(i)=1|D(i)=0)\\
 &=P_r(P_t|\frac{1}{2}\mu^2gh+\mu\alpha_{st}\alpha_{sr}\alpha_{tr}|)\\
 &=Q({\frac{\mu^2{g}{h}P_t\sqrt{N_s}}{2\sqrt{2}(N_0+N_1)}})=\frac{1}{2}erfc(\frac{\mu^2{g}{h}P_t\sqrt{N_s}}{4({\delta_0}^2+{\delta_1}^2)}).
\end{aligned}
\end{equation}
where $Q(\cdot)$ is the Gaussian $Q$-function, which is defined as
\begin{equation}
\begin{aligned}
Q(x) &=\frac{1}{\sqrt{2\pi}}\int_{x}^{\infty}exp(-\frac{t^2}{2})dt.\\
\end{aligned}
\end{equation}


\end{document}